\documentclass[runningheads]{llncs}
\usepackage[T1]{fontenc}
\usepackage{amssymb}
\usepackage{amsmath}
\usepackage{bm}
\usepackage[dvipdfmx]{graphicx}
\usepackage{caption}

\usepackage{algorithmic}
\usepackage{algorithm}
\usepackage{float} 
\usepackage{subcaption}
\usepackage{cases}

\begin{document}
\title{An Algorithm for Computing the Leading Monomials of a Minimal Gröbner Basis of Generic Sequences}
\titlerunning{An algorithm for computing the leading monomials of generic sequences}
%
\author{Kosuke Sakata\inst{1} \and
Tsuyoshi Takagi\inst{1}}
\authorrunning{K. Sakata et al.}
%
\institute{Department of Mathematical Informatics, The University of Tokyo, Japan
\email{\{kosuke-sakata-rb,takagi\}@g.ecc.u-tokyo.ac.jp}
}
\maketitle              
%

    

\begin{abstract}

We present an efficient algorithm for computing the leading monomials of a minimal Gröbner basis of a generic sequence of homogeneous polynomials. Our approach bypasses costly polynomial reductions by exploiting structural properties conjectured to hold for generic sequences—specifically, that their leading monomial ideals are weakly reverse lexicographic and that their Hilbert series follow a known closed-form expression. The algorithm incrementally constructs the set of leading monomials degree by degree by comparing Hilbert functions of monomial ideals with the expected Hilbert series of the input ideal. To enhance computational efficiency, we introduce several optimization techniques that progressively narrow the search space and reduce the number of divisibility checks required at each step. We also refine the loop termination condition using degree bounds, thereby avoiding unnecessary recomputation of Hilbert series. 
Experimental results confirm that the proposed method substantially reduces both computation time and memory usage compared to conventional Gröbner basis computations for computing the leading monomials of a minimal Gröbner basis of generic sequences. 


\end{abstract}

\section{Introduction}\label{sec:intro}

Gröbner basis computation is a central tool in computational algebra, with wide applications in solving systems of polynomial equations and analyzing ideals in polynomial rings. Despite advances in algorithms such as Buchberger's algorithm~\cite{Buch,BC} and the $F_4$/$F_5$ algorithms~\cite{F4,F5}, the cost of computing Gröbner bases remains high, especially for systems with many variables or higher-degree inputs. A substantial portion of this cost comes from the need for extensive polynomial reductions and S-polynomial computations.

When the input system satisfies certain genericity assumptions, its algebraic structure becomes more tractable. In particular, it is conjectured that a generic sequence of homogeneous polynomials yields a weakly reverse-lexicographic leading monomial ideal under the graded reverse lexicographic order~\cite{Pardue,MS}. This conjecture, originally due to Moreno--Socías, implies an explicit formula for the Hilbert series of the ideal, and allows the structure of the Gröbner basis to be predicted without full reduction.

Motivated by this observation, we propose an algorithm named LGB for computing the leading monomials of a minimal Gröbner basis of a generic sequence. The design of LGB is inspired by Hilbert-driven strategies, which are known for predicting the number and degree of basis elements using Hilbert series.  LGB iteratively constructs the set of leading monomials degree by degree by comparing the Hilbert function of the constructed monomial ideal to the target Hilbert series derived from the input.

To further improve efficiency, we introduce a sequence of refinements: (i) replacing the termination condition with a degree-based bound, (ii) reducing the number of monomials that must be tested for membership, and (iii) minimizing the number of leading monomials used for divisibility checks.

This paper is structured as follows. Section~\ref{sec:pre} reviews basic concepts including polynomial rings, Gröbner bases, Hilbert series, and generic sequences. In Section~\ref{sec:LGB}, we introduce the LGB algorithm and prove its correctness and termination, and a toy example is presented. Sections~\ref{sec:d},\ref{sec:M},\ref{sec:L} describe successive improvements to the algorithm, targeting termination, monomial filtering, and generator reduction. Section~\ref{sec:imp} presents the final improved version of LGB and validates its performance through experiments for computing the leading monomials of a minimal Gröbner basis of generic sequences. 

\section{Preliminary}\label{sec:pre}
\subsection{Polynomial Ring}\label{sub:poly}

Let \( \mathbb{N} \) denote the set of non-negative integers and let \( k \) be a field. Define \( R = k[x_1, \ldots, x_n] \) as the polynomial ring in \( n \) variables over \( k \). Every polynomial \( f \in R \) can be uniquely written as a finite sum
$
f = \sum_{\alpha \in \mathbb{N}^n} c_\alpha x^\alpha = \sum_{\alpha \in \mathbb{N}^n} c_\alpha x_1^{\alpha_1} \cdots x_n^{\alpha_n},
$
with \( c_\alpha \in k \). Each such expression \( c_\alpha x^\alpha \) is called a \emph{term}, and \( x^\alpha \) is called a \emph{monomial}.

We denote by \( M \) the set of all monomials in \( R \), i.e.,
$
M = \left\{ x^\alpha \mid \alpha \in \mathbb{N}^n \right\}.
$
For any integer \( d \geq 0 \), we write \( M_d = \{ m \in M \mid \deg(m) = d \} \) for the set of all monomials of total degree \( d \). The total degree of a polynomial \( f \in R \) is denoted by \( \deg(f) \).

Given a set of polynomials \( F \subset R \), we write \( \langle F \rangle \) for the ideal generated by \( F \). An ideal is called a \emph{monomial ideal} if it admits a generating set consisting solely of monomials.

In the theory of Gröbner bases, a \emph{term order} plays a central role. A term order is a total order on \( M \). Throughout this paper, we use the \emph{graded reverse lexicographic order} (grevlex), defined as follows: for two monomials
$
x^\alpha = x_1^{\alpha_1} \cdots x_n^{\alpha_n}, x^\beta = x_1^{\beta_1} \cdots x_n^{\beta_n},
$
we declare \( x^\alpha < x^\beta \) if either \( \deg(x^\alpha) < \deg(x^\beta) \), or \( \deg(x^\alpha) = \deg(x^\beta) \) and, reading indices from right to left, the first index \( i \) such that \( \alpha_i \neq \beta_i \) satisfies \( \alpha_i < \beta_i \).

The \emph{leading term} \( \mathrm{LT}(f) \) of a polynomial \( f \in R \) is defined to be the term of \( f \) that is largest with respect to the chosen term order. The monomial part of the leading term is denoted by \( \mathrm{LM}(f) \), called the leading monomial. A polynomial is said to be \emph{monic} if the coefficient of its leading term is equal to 1.
Given a finite set of polynomials $F = \{f_1, \ldots, f_m\} \subset R$, we define the leading monomial set $\mathrm{LM}(F)$ as the set of leading monomials of its elements, that is,
$
\mathrm{LM}(F) := \{\mathrm{LM}(f_i) \mid 1 \leq i \leq m \}.
$
\subsection{Gröbner Bases}\label{sub:GB}

Let \( I \subset R \) be an ideal, and let \( G = \{g_1, \dots, g_s\} \subset I \) be a finite set of non-zero polynomials.

\begin{definition}
The set \( G \) is called a \emph{Gröbner basis} of \( I \) if, for every non-zero polynomial \( f \in I \), there exists \( g \in G \) such that \( \mathrm{LM}(g) \mid \mathrm{LM}(f) \).
\end{definition}

A Gröbner basis \( G \) is said to be \emph{minimal} if all elements of \( G \) are monic, and for any distinct pair \( g, g' \in G \), the leading monomial \( \mathrm{LM}(g') \) does not divide \( \mathrm{LM}(g) \).
Furthermore, \( G \) is called a \emph{reduced Gröbner basis} if it is minimal and, in addition, for every \( g \in G \) and for every term \( r \) occurring in \( g \), no leading monomial from \( G \setminus \{g\} \) divides \( r \).

\begin{definition}
Let \( I \subset R \) be a homogeneous ideal and \( G \subset I \) a Gröbner basis. For any non-negative integer \( d \), define
\[
G_{\le d} := \{ g \in G \mid \deg(g) \le d \}.
\]
Then \( G_{\le d} \) is called a \emph{Gröbner basis up to degree \( d \)}. If \( G \) is minimal, then \( G_{\le d} \) is called a \emph{minimal Gröbner basis up to degree \( d \)}.
\end{definition}




\subsection{Hilbert Series}

Every polynomial \( f \in R \) admits a unique decomposition into homogeneous components:
$
f = f_0 + f_1 + \cdots + f_d, \quad \text{where } f_i \in R_i,
$
with \( R_i \) denoting the \( k \)-vector space of homogeneous polynomials of degree \( i \). Accordingly, the ring \( R \) admits a graded decomposition:
$
R = \bigoplus_{i \geq 0} R_i.
$

Given a homogeneous ideal \( I \subset R \), define \( I_i := I \cap R_i \). Then \( I \) inherits a graded structure:
$
I = \bigoplus_{i \geq 0} I_i, \text{and } R/I = \bigoplus_{i \geq 0} (R/I)_i, \text{with } (R/I)_i \cong R_i / I_i.
$

\begin{definition}[Hilbert function and Hilbert series]
Let \( I \subset R \) be a homogeneous ideal. The Hilbert function of \( R/I \) is defined by
$
h_{R/I}(i) := \dim_k (R/I)_i 
$
for $i \geq 0$.
The associated generating function
$
H_{R/I}(z) := \sum_{i \geq 0} h_{R/I}(i)\, z^i
$
is called the Hilbert series of \( R/I \).
\end{definition}

\begin{proposition}[\cite{CLO}]
Let \( I \subset R \) be an ideal, and let \( G \subset I \) be a finite set of polynomials. Then, for all \( d \geq 0 \),
\[
h_{R/I}(d) \leq h_{R/\langle \mathrm{LM}(G) \rangle}(d).
\]
Moreover, \( G \) is a Gröbner basis of \( I \) if and only if equality holds for all \( d \geq 0 \).
\end{proposition}

\begin{proposition}[\cite{CLO}]
Suppose \( I \subset R \) is a homogeneous ideal and \( G \subset I \) is a finite set. Then the following statements hold:
\begin{enumerate}
    \item The truncated Hilbert series \( H_{R/I}(z) \) and \( H_{R/\langle \mathrm{LM}(G) \rangle}(z) \) agree up to degree \( d \) if and only if \( G_{\le d} \) is a Gröbner basis of \( I \) up to degree \( d \).
    \item The equality \( H_{R/I}(z) = H_{R/\langle \mathrm{LM}(G) \rangle}(z) \) holds if and only if \( G \) is a Gröbner basis of \( I \).
\end{enumerate}
\end{proposition}

\begin{remark}
Let \( G \subset R \) be a finite set of polynomials and define \( L_G := \mathrm{LM}(G) \). If
\[
h_{R/I}(d) = h_{R/\langle L_G \rangle}(d) \quad \text{for all } d \geq 0,
\]
then \( L_G \) coincides with the set of leading monomials of a Gröbner basis of \( I \). Furthermore, if no strict divisibility relations exist among the elements of \( L_G \), then \( L_G \) is exactly the set of leading monomials of a \emph{minimal} Gröbner basis of \( I \).
\end{remark}
\subsection{Computation of Hilbert Series}

To compute the Hilbert series of quotient rings of the form \( R/J \), where \( J \subset R \) is a monomial ideal, we employ a recursive procedure known as HPS~\cite{CLO}. This subroutine evaluates the series \( H_{R/J}(z) \) by analyzing the minimal generators of \( J \).

Let \( T \) be the minimal set of monomial generators of \( J \). The algorithm proceeds according to the following cases:
\begin{itemize}
    \item \textbf{Case 1:} If \( T = \emptyset \), then \( J = (0) \) and the Hilbert series is
    $
    H_{R/J}(z) = \frac{1}{(1 - z)^n}.
    $
    \item \textbf{Case 2:} If \( T = \{1\} \), then \( J = R \) and the quotient is zero, so
    $
    H_{R/J}(z) = 0.
    $
    \item \textbf{Case 3:} If \( T \) consists of \( p \) variables \( \{x_{i_1}, \dots, x_{i_p}\} \), then
    $
    H_{R/J}(z) = \frac{1}{(1 - z)^{n - p}}.
    $

    \item \textbf{Case 4:} In all other cases, the algorithm selects a monomial \( m \in T \) of degree at least two, chooses a variable \( x_i \) dividing \( m \), and recursively computes:
    $
    H_{R/J}(z) = H_{R/(J + \langle x_i \rangle)}(z) + z \cdot H_{R/(J : x_i)}(z).
    $
\end{itemize}
To understand the recursive step in Case 4, we briefly recall the definition of a \emph{colon ideal}. Given two ideals \( A, B \subset R \), the colon ideal \( A : B \) is defined as
$
A : B := \{ f \in R \mid f B \subset A \}.
$
In particular, for a variable \( x_i \), the colon ideal \( J : x_i \) consists of all polynomials \( f \in R \) such that \( f x_i \in J \). In the context of monomial ideals, this operation corresponds to removing \( x_i \) from the monomials of \( J \) whenever it appears as a factor.

This recursive process terminates when all generators are reduced to variables or constants. The HPS routine is used repeatedly in later sections to evaluate Hilbert series of monomial ideals generated during intermediate steps of the proposed algorithms.

\begin{algorithm}[t]
    \begin{algorithmic}[1]    
    \caption{{\bf HPS}}
    \label{HPS}
        \REQUIRE a monomial ideal $J \subseteq k[x_1,...,x_n]$
        \ENSURE Hilbert series $H_{R/J}(z)$
        \STATE $T \leftarrow $ minimal set of generators for $J$
        \IF{$T=\varnothing$}
            \STATE $H \leftarrow \frac{1}{(1-z)^n}$
        \ELSIF{$T=\{1\}$}
            \STATE $H \leftarrow 0$
        \ELSIF{$T$ consists of $p$ monomials of degree $1$}
            \STATE $H \leftarrow \frac{1}{(1-z)^{n-p}}$
        \ELSE
            \STATE Select $x_i$ appearing in a monomial of degree $> 1$ in $T$
            \STATE $H \leftarrow {\bf HPS}(J+\langle x_i \rangle) + z \cdot {\bf HPS}(J:x_i)$
        \ENDIF
        \RETURN $\{ H \}$
    \end{algorithmic}
\end{algorithm}

\subsection{Generic Sequences}






When polynomial sequences are chosen ``at random'', they are typically expected to exhibit certain generic properties with high probability. Informally, a sequence is said to be \emph{generic} if its coefficients are chosen without specific algebraic dependencies.

Let \( F = \{f_1, \dots, f_m\} \subset R = k[x_1, \dots, x_n] \) be a sequence of non-zero homogeneous polynomials, and let \( I = \langle F \rangle \) be the ideal they generate. It is widely believed—and supported by empirical evidence and theoretical considerations—that generic sequences give rise to highly structured leading monomial ideals. In particular, the following conjecture plays a central role in our discussion.

\begin{definition}[Weakly reverse-lexicographic ideal]
A monomial ideal \( J \subset R \) is called \emph{weakly reverse-lexicographic} if, for every minimal generator \( r \in J \), and for every monomial \( r' \) of the same total degree with \( r' > r \) in graded reverse lexicographic order, we have \( r' \in J \).
\end{definition}

\begin{conjecture}[Moreno–Socías~\cite{Pardue,MS}]
Let \( F \subset R \) be a generic sequence of homogeneous polynomials, and let \( I = \langle F \rangle \). Then the leading-term ideal \( \mathrm{LT}(I) \) is weakly reverse-lexicographic.
\end{conjecture}

This conjecture has important implications for the Hilbert series of the ideal \( I \). To formalize these implications, we introduce the following truncation operator on formal power series.

\begin{definition}
Let \( h(z) = \sum_{i \ge 0} a_i z^i \) be a formal power series with \( a_i \in \mathbb{Z} \). Define
$
[h(z)] := \sum_{i=0}^{k-1} a_i z^i,
$
where \( k = \min\{i \mid a_i \le 0\} \), and take \( k = \infty \) if all \( a_i > 0 \). In the latter case, we set \( [h(z)] = h(z) \).
\end{definition}

\begin{conjecture}[\cite{Pardue,Froberg}]
Let \( F = \{f_1, \dots, f_m\} \subset R \) be a generic sequence of homogeneous polynomials with degrees \( \deg(f_i) = d_i \). Then the Hilbert series of \( R/\langle F \rangle \) satisfies
\[
H_{R/\langle F \rangle}(z) = \left[ \frac{\prod_{i=1}^m (1 - z^{d_i})}{(1 - z)^n} \right].
\]
\end{conjecture}

\begin{theorem}[\cite{Pardue}]
If the Moreno–Socías conjecture holds, then the formula in the Hilbert series conjecture is valid for all generic sequences.
\end{theorem}

In the remainder of this paper, we proceed under the assumption that the Moreno–Socías conjecture is true, and hence that the compact Hilbert series formula above applies to every generic sequence.

\section{Our Algorithm for the Leading Monomials of a Minimal Gröbner Basis of a Generic Sequence}\label{sec:LGB}


In this section, we present an algorithm to compute the leading monomials of a minimal Gröbner basis for a generic sequence. In general, obtaining the leading monomials of a Gröbner basis requires executing a complete Gröbner basis algorithm, which is computationally intensive. However, under the assumption that the Moreno–Socías conjecture holds, the structure of a generic sequence allows for an efficient computation of the leading monomials without performing polynomial reductions.

Specifically, if the sequence \( F = \{f_1, \dots, f_m\} \subset R \) is generic, then its leading monomial ideal \( \mathrm{LM}(F) \) is weakly reverse-lexicographic. In this case, the Hilbert series of \( R / \langle F \rangle \) completely determines the number of leading monomials of each degree in a minimal Gröbner basis. Therefore, one can construct the correct leading monomials degree-by-degree by matching the Hilbert function of the currently constructed monomial ideal with that of the true ideal.

We denote this algorithm by LGB, and present it below as Algorithm~\ref{alg:LGB}. 
The design of LGB is inspired by Hilbert-driven algorithm (designed in~\cite{CLO}), which are known for predicting the number and degree of basis elements using Hilbert series. The input is the number of variables \( n \), the number of polynomials \( m \), and the degree list \( \{d_i\}_{i=1}^{m} \) such that \( \deg(f_i) = d_i \). The output is the set \( L_G \) of leading monomials of a minimal Gröbner basis of the generic sequence \( F \).

The algorithm operates as follows. Line 1 computes the Hilbert series \( H_{R/\langle F \rangle}(z) \) using the formula for generic sequences. Line 2 initializes the leading monomials \( L_G \) to the empty set. Line 3 computes the Hilbert series \( H_{R/\langle L_G \rangle}(z) \) of the monomial ideal generated by \( L_G \), using the HPS algorithm described in Section~2.4. Line 4 initializes the degree counter \( d \). The main loop (lines 5–10) iteratively increments \( d \) and updates \( L_G \) until the Hilbert series of the constructed monomial ideal matches that of the original ideal.

In each iteration, the set \( B_d \) of degree-\( d \) monomials not divisible by any element of \( L_G \) is determined (line 6). The difference between \( \#B_d \) and the true Hilbert function value \( h_{R/\langle F \rangle}(d) \) yields the number \( N_d \) of new leading monomials to be added (line 7). The \( N_d \) largest monomials in \( B_d \) with respect to the term order are selected and added to \( L_G \) (line 8). The updated Hilbert series is computed again in line 9. Once the two Hilbert series agree, the loop terminates, and the algorithm returns the complete set of leading monomials.

\begin{algorithm}[t]
    \begin{algorithmic}[1]    
    \caption{{\bf LGB}}
    \label{alg:LGB}
        \REQUIRE $n \in \mathbb{N}$, $m \in \mathbb{N}$, $\{d_i\}_{i=1}^{m} \in \mathbb{N}^m$
        \ENSURE A set of Leading monomials $L_G$ of minimal Gr\"obner basis for a generic sequence $F = \{f_1,...,f_m\}$ with $n$ variables, $m$ polynomials and $\mathrm{deg}(f_i)=d_i$
        \STATE $H_{R/\langle F \rangle}(z) \leftarrow \left[ \frac{\prod^{m}_{i=1}(1-z^{d_i})}{(1-z)^n} \right]$
        \STATE $L_G \leftarrow \varnothing$
        \STATE $H_{R/\langle L_G \rangle}(z) \leftarrow  {\bf HPS}(\langle L_G \rangle)$
        \STATE $d \leftarrow 0$
        \WHILE{$H_{R/\langle F \rangle}(z) \neq H_{R/ \langle L_G \rangle}(z)$}
        \STATE $B_d \leftarrow \{ m \in M_d \mid m \notin \langle L_G \rangle\}$
        \STATE $N_d \leftarrow \#B_d - h_{R/\langle F \rangle}(d)$ 
        \STATE $L_G \leftarrow L_G \cup \{ $the $N_d$ largest elements in term order from $B_d\}$ 
        \STATE $H_{R/\langle L_G \rangle}(z) \leftarrow {\bf HPS}(\langle L_G \rangle)$
        \STATE $d \leftarrow d+1$
        \ENDWHILE
        \RETURN $\{ L_G \}$
    \end{algorithmic}
\end{algorithm}

\begin{theorem}\label{thm:LGB_correctness}
The set \(L_G\) returned by the algorithm {\normalfont LGB} is exactly the set of leading monomials of a minimal Gröbner basis of a generic sequence.
\end{theorem}

\begin{proof}
We prove by induction on the degree parameter \(d\) that, at the end of the loop for degree \(d\), the set \(L_G\) coincides with the leading monomials of a minimal Gröbner basis up to degree \(d\).

Assume inductively that, at the beginning of the iteration for degree \(d\), the set \(L_G\) already contains the leading monomials of a minimal Gröbner basis of all degrees strictly less than \(d\).
Let \(B_d\) be the set of degree-\(d\) monomials not divisible by any element of \(L_G\) (line 6 of the algorithm).
Because a monomial not divisible by \(L_G\) represents a non-zero element of the quotient \(R/\langle L_G\rangle\), the cardinality of \(B_d\) exceeds the Hilbert function value \(h_{R/\langle F\rangle}(d)\) by exactly the number of minimal generators of degree \(d\).
Hence \(N_d=\#B_d-h_{R/\langle F\rangle}(d)\) equals the number of leading monomials of degree \(d\).

Denote by \(L_{G,d}\subset M_d\) the true set of leading monomials of degree \(d\) in a minimal Gröbner basis and by \(L_{G,d}'\subset M_d\) the set chosen by the algorithm, that is, the \(N_d\) largest monomials in \(B_d\).
Since \(L_{G,d}\subset B_d\) (otherwise a monomial of \(L_{G,d}\) would be divisible by a generator of lower degree, contradicting minimality) and \(\#L_{G,d}=\#L_{G,d}'=N_d\), it suffices to show \(L_{G,d}=L_{G,d}'\).

Assume to the contrary that \(L_{G,d}\neq L_{G,d}'\).
Then there exist monomials \(l\in L_{G,d}\setminus L_{G,d}'\) and \(l'\in L_{G,d}'\setminus L_{G,d}\).
Because \(L_{G,d}'\) consists of the largest elements of \(B_d\), we have \(l'<l\) in the term order.
The monomial \(l'\) is minimal in the leading-term ideal, so \(l'\) is not divisible by any other minimal generator of the same degree.
Since \(l\) and \(l'\) have the same degree and \(l'\nmid l\), the presence of \(l\) contradicts the weakly reverse-lexicographic property expected for a generic sequence.
Therefore the assumption is impossible, and \(L_{G,d}=L_{G,d}'\).

Consequently, after line 8 the set \(L_G\) contains the correct leading monomials up to degree \(d\).
When the loop ends at some degree \(d'\) with \(H_{R/\langle F\rangle}(z)=H_{R/\langle L_G\rangle}(z)\), no further generators of higher degree can exist; otherwise the Hilbert series would differ at that degree.
Thus, on exit the set \(L_G\) is precisely the set of leading monomials of a minimal Gröbner basis of the generic sequence. \qed
\end{proof}

\begin{theorem}\label{thm:terminate}
The algorithm \textnormal{LGB} terminates after finite iterations.
\end{theorem}

\begin{proof}
The goal of LGB is to compute the set \( L_G \) of leading monomials of a minimal Gröbner basis for a generic sequence \( F \subset R = k[x_1, \dots, x_n] \). Equivalently, it aims to determine a minimal generating set of the monomial ideal \( \langle \operatorname{LT}(F) \rangle \).

By Dickson's lemma, every monomial ideal in a polynomial ring over a field is finitely generated. Therefore, the ideal \( \langle \operatorname{LT}(F) \rangle \) admits a finite minimal generating set. Let \( D \) be a positive integer such that all minimal generators of \( \langle \operatorname{LT}(F) \rangle \) have degree at most \( D \).

In Algorithm~\ref{alg:LGB}, the monomials in \( L_G \) are constructed incrementally by degree. Once the degree counter \( d \) reaches \( D \), all generators of \( \langle \operatorname{LT}(F) \rangle \) have been recovered, and hence
$
\langle L_G \rangle = \langle \operatorname{LT}(F) \rangle.
$
Consequently, the Hilbert series of the two ideals coincide:
$
H_{R/\langle L_G \rangle}(z) = H_{R/\langle F \rangle}(z),
$
and the while-loop condition in Algorithm~\ref{alg:LGB} is no longer satisfied.

Therefore, the algorithm terminates after finite iterations. \qed
\end{proof}

\subsection{Toy Example}

In this section, we demonstrate a concrete example of computing the leading monomials of a minimal Gr\"obner basis of a generic sequence $F$ using the LGB algorithm. Consider the case with $n = 3$, $m = 4$, and degrees $\{d_i\}_{i=1}^{4} = \{2, 2, 3, 4\}$.

Assuming the conjecture on the Hilbert series of generic sequences is valid, the Hilbert series of $R/\langle F \rangle$ is given by
\[
H_{R/\langle F \rangle}(z) = \left[ \frac{(1 - z^2)^2 (1 - z^3)(1 - z^4)}{(1 - z)^3} \right] = 3z^3 + 4z^2 + 3z + 1.
\]

\begin{itemize}
\item[$\bullet$] Degree $d = 0$:\\
We have $M_0 = \{1\}$ and $L_G = \emptyset$, so $B_0 = \{1\}$. Thus, $N_0 = \#B_0 - h_{R/\langle F \rangle}(0) = 1 - 1 = 0$, and no monomials are added to $L_G$.

\item[$\bullet$] Degree $d = 1$:\\
We have $M_1 = \{x_1, x_2, x_3\}$ and $L_G = \emptyset$, so $B_1 = \{x_1, x_2, x_3\}$. Hence, $N_1 = \#B_1 - h_{R/\langle F \rangle}(1) = 3 - 3 = 0$, and again no monomials are added.

\item[$\bullet$] Degree $d = 2$:\\
We have $M_2 = \{x_1^2, x_1x_2, x_2^2, x_1x_3, x_2x_3, x_3^2\}$ and $L_G = \emptyset$, so $B_2 = M_2$. Thus, $N_2 = 6 - 4 = 2$. The two largest monomials in $B_2$ with respect to the term order are $\{x_1^2, x_1x_2\}$. Therefore, $L_G = \{x_1^2, x_1x_2\}$.

\item[$\bullet$] Degree $d = 3$:\\
Using the current $L_G$, we compute the Hilbert series $H_{R/\langle L_G \rangle}(z)$, which yields $1 + 3z + 4z^2 + 5z^3 + \cdots$. Since this differs from $H_{R/\langle F \rangle}(z)$, we proceed. The remaining degree-3 monomials not divisible by $L_G$ are $B_3 = \{x_2^3, x_2^2x_3, x_1x_3^2, x_2x_3^2, x_3^3\}$. Thus, $N_3 = 5 - 3 = 2$, and we select the two largest monomials: $\{x_2^3, x_2^2x_3\}$. Hence, $L_G = \{x_1^2, x_1x_2, x_2^3, x_2^2x_3\}$.

\item[$\bullet$] Degree $d = 4$:\\
Now, compute $H_{R/\langle L_G \rangle}(z) = 1 + 3z + 4z^2 + 3z^3 + 3z^4 + \cdots$. As this is not equal to $H_{R/\langle F \rangle}(z)$, we proceed. Let $B_4 = \{x_1x_3^3, x_2x_3^3, x_3^4\}$, then $N_4 = 3 - 0 = 3$, and the three largest monomials are $\{x_1x_3^3, x_2x_3^3, x_3^4\}$. Thus, we obtain
\[
L_G = \{x_1^2, x_1x_2, x_2^3, x_2^2x_3, x_1x_3^3, x_2x_3^3, x_3^4\}.
\]
\end{itemize}

Finally, compute $H_{R/\langle L_G \rangle}(z) = 3z^3 + 4z^2 + 3z + 1$, which coincides with $H_{R/\langle F \rangle}(z)$, concluding the computation. Therefore, the above $L_G$ is the set of leading monomials of a minimal Gr\"obner basis for the generic sequence $F$.

\section{Improvements to LGB by Modifying the Loop Termination Condition}\label{sec:d}

In this section, we propose an enhancement of the LGB algorithm by modifying its loop termination condition. Let $F = \{f_1, \ldots, f_m\}$ be a generic sequence with $n$ variables and $m$ polynomials such that $\deg(f_i) = d_i$.

\subsection{Modification of the Loop Termination Condition When $n \leq m$}

In Algorithm~2 (LGB), the Hilbert series $H_{R/\langle L_G \rangle}(z)$ must be computed using the HPS algorithm at each iteration (line 9). However, as the number of elements in $L_G$ increases, the computation cost of HPS also increases significantly. If possible, this computation should be avoided. To address this, we consider an alternative termination condition that eliminates the need to compute HPS.

When $n \leq m$, the Hilbert series of $F$ is finite. Let $D$ be one plus the degree of $H_{R/\langle F \rangle}(z)$; that is, $D$ is the smallest integer such that
$
\langle F \rangle_d = R_d \text{ for all } d \geq D.
$
This value $D$ is known as the \emph{degree of regularity}, which is widely used in estimating the computational complexity of multivariate polynomial cryptosystems.

By the definition of the Hilbert series, we have $\dim(R/\langle F \rangle)_d = 0$ for all $d \geq D$. Therefore, all elements in $R_D$ are reducible by Gr\"obner basis of $\langle F \rangle$. If there existed an element of a minimal Gr\"obner basis of $\langle F \rangle$ with degree $d > D$, it would imply the existence of a monomial of degree $d$ that is not reducible by degree $\leq D$ monomials, contradicting the minimality of the basis.

Hence, when $n \leq m$, the leading monomials of a minimal Gr\"obner basis of $F$ do not increase beyond degree $D$, and it suffices to compute LGB up to degree $D$.

\subsection{Modification of the Loop Termination Condition When $n > m$}

In the case where $n > m$, unlike the case $n \leq m$, the Hilbert series of $F$ is infinite. Therefore, the same approach as in the previous section cannot be applied. Instead, we can utilize the well-known Macaulay bound as an upper limit for the computation.

The Macaulay bound provides an upper bound on the degree $D$ as follows:
$
D \leq \sum_{i=1}^{m} (d_i - 1) + 1.
$

Thus, when $n > m$, we set the maximum degree $D$ for which the LGB algorithm computes to be the right-hand side of the above inequality. This allows us to establish a guaranteed termination condition for LGB in the case of $n > m$ without computing the Hilbert series at every iteration.

\begin{remark}
The Macaulay bound provides an upper bound on the degree $D$ required to compute all leading monomials of a minimal Gröbner basis when $n > m$. While this bound is not necessarily tight in general, in all our computational experiments, the maximal degree $d$ for which $N_d \neq 0$ was 
$
d = \sum_{i=1}^{m} (d_i - 1) + 1
$. 
This suggests that the Macaulay bound may in practice suffice to compute the complete set of leading monomials in LGB.
\end{remark}
\section{Improvement of LGB by Reducing the Number of Monomials to be Checked for Divisibility by $\langle L_G \rangle$}\label{sec:M}

In the LGB algorithm, one of the most computationally expensive steps is determining whether a given monomial of degree $d$ belongs to the ideal $\langle L_G \rangle$ (line 6 in Algorithm~2). This requires checking all monomials of degree $d$, whose number is
$
\binom{n + d - 1}{d}.
$
As the degree of a minimal Gr\"obner basis of the generic sequence increases, the number of such monomials increases rapidly, leading to a significant rise in computational cost.

In this section, we introduce methods to reduce the number of monomials that must be checked for membership in $\langle L_G \rangle$, thereby improving the efficiency of the algorithm. We present multiple techniques (referred to as {\bf tier 1}, {\bf tier 2}, and {\bf tier 3}) in a step-by-step manner. Each subsequent tier refines the previous one and achieves greater reduction in the number of monomials. Among them, the method described last provides the most significant reduction.

Let $L^{(d)}_G$ denote the set of leading monomials of degree $d$ in a minimal Gr\"obner basis. In the LGB algorithm, this corresponds to the update of $L_G$ performed at line 8 in the iteration for parameter $d$. Let $B_d$ be the set of degree $d$ monomials not divisible by any element in $\langle L^{(d-1)}_G \rangle$, i.e., the monomials of $M_d$ that are not in the ideal generated by leading monomials up to degree $d-1$.

Let $\widetilde{B}_d$ denote the set of monomials in $M_d$ that are not divisible by the updated ideal $\langle L^{(d)}_G \rangle$, and let $\widetilde{M}_d$ be the subset of $M_d$ whose elements need to be tested for membership in $\langle L^{(d-1)}_G \rangle$ when computing $B_d$. The goal of this section is to reduce the size of $\widetilde{M}_d$.

We observe that $\widetilde{B}_d$ can be computed from $B_d$ and $N_d$ (the number of new leading monomials added at degree $d$) as follows:
\[
\widetilde{B}_d \leftarrow B_d \setminus \{\text{the } N_d \text{ largest monomials (in term order) in } B_d\}.
\]
This formula follows from the definitions. $B_d$ consists of monomials in $M_d$ not divisible by $\langle L^{(d-1)}_G \rangle$, and $\widetilde{B}_d$ consists of those not divisible by $\langle L^{(d)}_G \rangle$. The only difference between $L^{(d-1)}_G$ and $L^{(d)}_G$ is the set of newly added monomials of degree $d$ (added in line 8). Hence, the difference between $B_d$ and $\widetilde{B}_d$ lies only in those added monomials. Removing those from $B_d$ yields $\widetilde{B}_d$.

\subsection{tier 1}\label{sub:tier1}

Based on Theorem~4, $B_d$ can be computed as follows:
\begin{align*}
\widetilde{M}_d &\leftarrow \{ x_i b \mid i \in \{1, \ldots, n\},\ b \in \widetilde{B}_{d-1} \}, \\
B_d &\leftarrow \{ m \in \widetilde{M}_d \mid m \notin \langle L^{(d-1)}_G \rangle \}.
\end{align*}

\begin{theorem}
Let $\widetilde{M}_d = \{x_i b \mid i \in \{1, \ldots, n\},\ b \in \widetilde{B}_{d-1}\}$. Then,
$
B_d = \{ m \in \widetilde{M}_d \mid m \notin \langle L^{(d-1)}_G \rangle \}.
$
\end{theorem}

\begin{proof}
By definition,
$
B_d = \{ m \in M_d \mid m \notin \langle L^{(d-1)}_G \rangle \}.
$
Note that $M_d$ can be written as
$
M_d = \{ r' \cdot r \mid r' \in M_1,\ r \in M_{d-1} \}.
$
If $r \in \langle L^{(d-1)}_G \rangle$, then for all $r' \in M_1$, we have $r' r \in \langle L^{(d-1)}_G \rangle$, and thus $r' r \notin B_d$. Therefore, when constructing $M_d$, any product involving an $r$ in $\langle L^{(d-1)}_G \rangle$ will be excluded from $B_d$.

Since $M_{d-1} \setminus \langle L^{(d-1)}_G \rangle = \widetilde{B}_{d-1}$, it suffices to test the monomials
$
\widetilde{M}_d = \{ x_i b \mid i \in \{1, \ldots, n\},\ b \in \widetilde{B}_{d-1} \}
$
for membership in $\langle L^{(d-1)}_G \rangle$ in order to compute $B_d. \qed$
\end{proof}

\subsection{tier 2}\label{sub:tier2}

According to Theorem~5, the computation of $B_d$ proceeds as follows:
\begin{align*}
b &\leftarrow \text{the largest monomial in term order from } \widetilde{B}_{d-1}, \\
x_t &\leftarrow \text{the smallest variable in } b, \\
\widetilde{M}_d &\leftarrow \{ x_i b \mid i \in \{t, \ldots, n\},\ b \in \widetilde{B}_{d-1} \}, \\
B_d &\leftarrow \{ m \in \widetilde{M}_d \mid m \notin \langle L^{(d-1)}_G \rangle \}.
\end{align*}
This method of computing $B_d$ is referred to as \textbf{tier 2}.

\begin{theorem}
Let $r$ be the largest monomial in $\widetilde{B}_{d-1}$ with respect to the term order, and let $x_t$ be the smallest variable appearing in $r$. Define
\[
\widetilde{M}_d = \{ x_i b \mid i \in \{t, \ldots, n\},\ b \in \widetilde{B}_{d-1} \}.
\]
Then,
\[
B_d = \{ m \in \widetilde{M}_d \mid m \notin \langle L^{(d-1)}_G \rangle \}.
\]
\end{theorem}

\begin{proof}
By Theorem~4, it is sufficient to consider the set
$
\{ x_i b \mid i \in \{1, \ldots, n\},\ b \in \widetilde{B}_{d-1} \}
$
when computing $B_d$.

Clearly, this set can be partitioned as:
$
\{ x_i b \mid i \in \{1, \ldots, t-1\},\ b \in \widetilde{B}_{d-1} \} \cup \{ x_i b \mid i \in \{t, \ldots, n\},\ b \in \widetilde{B}_{d-1} \}.
$

We claim that every monomial in the first set is either in $\langle L^{(d-1)}_G \rangle$ or also appears in the second set. Let $b' \in \widetilde{B}_{d-1}$. Since $r \geq b'$ and $\deg(r) = \deg(b')$, by the definition of graded reverse lexicographic order, $b'$ must contain at least one variable $x_p$ such that $p \in \{t, \ldots, n\}$. Then $b' = x_p b''$ for some monomial $b''$.

Now, let $x_q \in \{x_i \mid i \in \{1, \ldots, t-1\}\}$, and suppose $x_q b' \notin \langle L^{(d-1)}_G \rangle$. Then,
$
x_q b' = x_q x_p b'' \notin \langle L^{(d-1)}_G \rangle,
$
implying $x_q b'' \notin \langle L^{(d-1)}_G \rangle$, i.e., $x_q b'' \in \widetilde{B}_{d-1}$. Therefore, $x_p x_q b'' \in \{ x_i b \mid i \in \{t, \ldots, n\},\ b \in \widetilde{B}_{d-1} \}$, showing that $x_q b'$ is already included or reducible via a term in the second set.

Hence, it suffices to consider only the set
\[
\widetilde{M}_d = \{ x_i b \mid i \in \{t, \ldots, n\},\ b \in \widetilde{B}_{d-1} \}
\]
for computing $B_d$. \qed
\end{proof}

\subsection{tier 3}\label{sub:tier3}

According to Theorem~6, $B_d$ can be computed as follows:
\begin{align*}
b &\leftarrow \text{the largest monomial in term order from } \widetilde{B}_{d-1}, \\
x_t &\leftarrow \text{the smallest variable in } b, \\
\widetilde{M}_d &\leftarrow \{ x_t b \mid b \in \widetilde{B}_{d-1} \}, \\
B_d &\leftarrow \{ m \in \widetilde{M}_d \mid m \notin \langle L^{(d-1)}_G \rangle \} \cup \{ x_i b \mid i \in \{t+1, \ldots, n\},\ b \in \widetilde{B}_{d-1} \}.
\end{align*}
We refer to this method of computing $B_d$ as \textbf{tier 3}.

\begin{theorem}
Let $r$ be the largest monomial in $\widetilde{B}_{d-1}$ with respect to the term order, and let $x_t$ be the smallest variable appearing in $r$. Define
\[
\widetilde{M}_d = \{ x_t b \mid b \in \widetilde{B}_{d-1} \}.
\]
Then,
\[
B_d = \{ m \in \widetilde{M}_d \mid m \notin \langle L^{(d-1)}_G \rangle \} \cup \{ x_i b \mid i \in \{t+1, \ldots, n\},\ b \in \widetilde{B}_{d-1} \}.
\]
\end{theorem}

\begin{proof}
By Theorem~5, the set of monomials that need to be tested for membership in $\langle L^{(d-1)}_G \rangle$ when computing $B_d$ is
$
\{ x_i b \mid i \in \{t, \ldots, n\},\ b \in \widetilde{B}_{d-1} \}.
$
Clearly, this set can be split as:
\[
\{ x_t b \mid b \in \widetilde{B}_{d-1} \} \cup \{ x_i b \mid i \in \{t+1, \ldots, n\},\ b \in \widetilde{B}_{d-1} \}.
\]

We now show that monomials in $\{ x_i b \mid i \in \{t+1, \ldots, n\},\ b \in \widetilde{B}_{d-1} \}$ do not belong to $\langle L^{(d-1)}_G \rangle$. Suppose, for contradiction, that $x_p b' \in \langle L^{(d-1)}_G \rangle$ for some $x_p \in \{ x_i \mid i \in \{t+1, \ldots, n\} \}$ and $b' \in \widetilde{B}_{d-1}$. Then there exists some $d' \leq d-1$ and some $l \in L^{(d')}_G$ with $\deg(l) = d'$ such that $l \mid x_p b'$.

Since $b' \in \widetilde{B}_{d-1}$, we have $b' \notin \langle L^{(d-1)}_G \rangle$, hence $l \nmid b'$. Let us write $l = x_1^{l_1} \cdots x_n^{l_n}$ and $b' = x_1^{b_1'} \cdots x_n^{b_n'}$. From $l \mid x_p b'$ and $l \nmid b'$, it follows that $l_p = b_p' + 1$ and $l_p \geq 1$. Therefore, $x_p$ must appear in $l$.

Since $l$ is a minimal generator of $L^{(d')}_G$ and we assume the conjecture on generic sequences, all monomials of degree $d'$ larger than $l$ (with respect to the term order) belong to $\langle L^{(d')}_G \rangle$. In particular, any monomial of degree $d'$ involving variables $x_i$ with $i < p$ must be in $\langle L^{(d')}_G \rangle$ due to the properties of graded reverse lexicographic order.

This implies that all monomials of the form $x_i c$ for $i \in \{1, \ldots, p-1\}$ and $c \in M_{d'-1}$ are in $\langle L^{(d')}_G \rangle \subseteq \langle L^{(d-1)}_G \rangle$. However, $r = x_t c \in \widetilde{B}_{d-1}$ for some $c \in M_{d-1}$ and $t < p$, which contradicts the assumption that $r \notin \langle L^{(d-1)}_G \rangle$. Thus, $x_p b' \notin \langle L^{(d-1)}_G \rangle$.

Therefore, the monomials $x_i b$ with $i \in \{t+1, \ldots, n\}$ and $b \in \widetilde{B}_{d-1}$ do not belong to $\langle L^{(d-1)}_G \rangle$ and can be directly included in $B_d$ without further checking.\qed
\end{proof}

\subsection{Effectiveness of Reducing $\widetilde{M}_d$ via Computational Experiments}\label{sub:exp}

In this subsection, we verify the effectiveness of the proposed methods for reducing the number of monomials $\widetilde{M}_d$ used in the computation of $B_d$, as introduced in Sections~\ref{sub:tier1}, \ref{sub:tier2}, and \ref{sub:tier3} (referred to as {\bf tier 1}, {\bf tier 2}, and {\bf tier 3}, respectively). We conduct experiments using the following two test cases:

\begin{itemize}
    \item \textbf{Case 1}: $n = 18$, $m = 19$, $d_i = 2$ for all $i = 1, \ldots, m$.
    \item \textbf{Case 2}: $n = 14$, $m = 10$, $d_i = 2$ for all $i = 1, \ldots, m$.
\end{itemize}

Table~\ref{tab:M} shows the number of elements in $\widetilde{M}_d$ at each degree $d$, comparing the results for tier 0 (the baseline), tier 1, tier 2, and tier 3.

In Case 1 of Table~\ref{tab:M}, we observe that the size of $\widetilde{M}_d$ increases rapidly with $d$ under {\bf tier 0}. In contrast, tier 3 achieves substantial reduction starting from $d = 1$, and the difference becomes even more pronounced as the degree increases. For instance, at $d = 10$, the number of monomials in {\bf tier 0} is $8,\!436,\!285$, while it is only $16,\!796$ in {\bf tier 3}, demonstrating a dramatic decrease. This result implies that the number of monomials requiring membership checks in $B_d$ is significantly reduced, greatly contributing to the overall computational efficiency of the algorithm.

Similarly, Case 2 in Table~\ref{tab:M} confirms that the size of $\widetilde{M}_d$ is markedly reduced using {\bf tier 3}. Compared to Case1 in Table~\ref{tab:M}, the reduction is even more noticeable in the higher-degree ranges. This suggests that the effect of {\bf tier 3} is especially beneficial in cases where $n > m$.

From these results, we experimentally confirm that the proposed reductions in $\widetilde{M}_d$ lead to a significant decrease in the number of monomials that must be examined during the computation of $B_d$. Consequently, the total computation time for determining the leading monomials of a minimal Gr\"obner basis of a generic sequence is greatly reduced.

\begin{table}[t]
  \centering
  \caption{Number of candidate monomials $\widetilde{M}_d$ at each degree $d$ for {\bf tier 0–3} methods.}
  \label{tab:M}
  \begin{minipage}[t]{0.48\textwidth}
    \centering
    \begin{tabular}{c|rrrr}
      \hline
      $d$ & \; {\bf tier 0} & \; {\bf tier 1} & \; {\bf tier 2} & \; {\bf tier 3} \\
      \hline
      1 & 18 & 18 & 18 & 1 \\
      2 & 171 & 171 & 171 & 18 \\
      3 & 1140 & 1095 & 906 & 152 \\
      4 & 5985 & 5064 & 3144 & 798 \\
      5 & 26334 & 17417 & 8794 & 2907 \\
      6 & 100947 & 45331 & 17864 & 7752 \\
      7 & 346104 & 89889 & 29640 & 15504 \\
      8 & 1081575 & 134447 & 31008 & 23256 \\
      9 & 3124550 & 145918 & 32946 & 25194 \\
      10 & 8436285 & 97791 & 16796 & 16796 \\
      \hline
    \end{tabular}
    \subcaption*{Case 1: $n=18,\ m=19,\ d_i=2$}
  \end{minipage}
  \hfill
  \begin{minipage}[t]{0.48\textwidth}
    \centering
    \begin{tabular}{c|rrrr}
      \hline
      $d$ & \; {\bf tier 0} & \; {\bf tier 1} & \; {\bf tier 2} & \; {\bf tier 3} \\
      \hline
      1 & 14 & 14 & 14 & 1 \\
      2 & 105 & 105 & 105 & 14 \\
      3 & 560 & 540 & 440 & 95 \\
      4 & 2380 & 2084 & 1442 & 420 \\
      5 & 8568 & 6364 & 3811 & 1375 \\
      6 & 27132 & 16049 & 8261 & 3598 \\
      7 & 77520 & 34670 & 15618 & 7937 \\
      8 & 203490 & 66257 & 27735 & 15360 \\
      9 & 497420 & 115040 & 45595 & 26880 \\
      10 & 1144066 & 185304 & 70473 & 43520 \\
      11 & 2496144 & 281344 & 103645 & 66304 \\
      \hline
    \end{tabular}
    \subcaption*{Case 2: $n=14,\ m=10,\ d_i=2$}
  \end{minipage}
\end{table}

\section{Improvement of LGB by Reducing the Number of Elements in $L_G$ Used in the Computation of $B_d$}\label{sec:L}

In the LGB algorithm, the computation of $B_d$ requires checking which monomials in a certain set (either $M_d$ or $\widetilde{M}_d$) are not contained in the ideal $\langle L_G \rangle$. As the number of elements in $L_G$ increases, the cost of searching for relevant divisors within $L_G$ grows substantially.

In this section, we focus on the monomial set $\widetilde{M}_d$ used in the {\bf tier 3} method (described in Section~\ref{sub:tier3}), and demonstrate that only a specific subset of $L_G$ is relevant for the divisibility checks. This observation allows us to reduce the number of $L_G$ elements that need to be examined, thereby improving the efficiency of $B_d$ computation.

\subsection{tier 4}\label{sub:tier4}

Based on Theorem~7, $B_d$ can be computed as follows:
\begin{align*}
b &\leftarrow \text{the largest monomial in term order from } \widetilde{B}_{d-1}, \\
x_t &\leftarrow \text{the smallest variable appearing in } b, \\
\widetilde{M}_d &\leftarrow \{ x_t b \mid b \in \widetilde{B}_{d-1} \}, \\
B_d &\leftarrow \{ m \in \widetilde{M}_d \mid m \notin \langle L^{(d-1)}_G \rangle \cap \langle x_t \rangle \} \cup \{ x_i b \mid i \in \{t+1, \ldots, n\},\ b \in \widetilde{B}_{d-1} \}.
\end{align*}

We refer to this method of computing $B_d$ as \textbf{tier 4}.

\begin{theorem}
Let $r$ be the largest monomial in $\widetilde{B}_{d-1}$ with respect to the term order, and let $x_t$ be the smallest variable appearing in $r$. Define
\[
\widetilde{M}_d = \{ x_t b \mid b \in \widetilde{B}_{d-1} \}.
\]
Then,
\[
B_d = \{ m \in \widetilde{M}_d \mid m \notin \langle L^{(d-1)}_G \rangle \cap \langle x_t \rangle \} \cup \{ x_i b \mid i \in \{t+1, \ldots, n\},\ b \in \widetilde{B}_{d-1} \}.
\]
\end{theorem}

\begin{proof}
Let $b' \in \widetilde{B}_{d-1}$ and $l \in L^{(d-1)}_G$. Suppose $l$ divides $x_t b'$. We show that $x_t$ must be one of the variables appearing in $l$.

Since $b' \in \widetilde{B}_{d-1}$, it follows that $b' \notin \langle L^{(d-1)}_G \rangle$, and thus $l \nmid b'$. Write $l = x_1^{l_1} \cdots x_n^{l_n}$ and $b' = x_1^{b'_1} \cdots x_n^{b'_n}$. The condition $l \mid x_t b'$ and $l \nmid b'$ implies that $l_t = b'_t + 1$, and thus $l_t \geq 1$. Therefore, $x_t$ must be one of the variables that appear in $l$.

Hence, for any $l \in L^{(d-1)}_G$ that divides a monomial in $\widetilde{M}_d = \{ x_t b \mid b \in \widetilde{B}_{d-1} \}$, it must be the case that $l \in \langle x_t \rangle$. That is, only those monomials in $L^{(d-1)}_G$ divisible by $x_t$ are relevant in checking whether an element of $\widetilde{M}_d$ lies in $\langle L^{(d-1)}_G \rangle$.

Therefore, when checking whether $m \in \widetilde{M}_d$ is in $\langle L^{(d-1)}_G \rangle$, it suffices to check only the elements of $L^{(d-1)}_G$ that lie in $\langle x_t \rangle$, i.e., in $\langle L^{(d-1)}_G \rangle \cap \langle x_t \rangle$. \qed
\end{proof}

\subsection{Effectiveness through Computational Experiments}\label{sub:exp2}

In this subsection, we experimentally verify the effectiveness of the {\bf tier 4} method described in Section~\ref{sub:tier4}, which reduces the number of elements in $L_G^{(d-1)}$ used during the computation of $B_d$ from $\widetilde{M}_d$. The following two cases, also used in Section~\ref{sub:exp}, are considered:

\begin{itemize}
    \item \textbf{Case 1}: $n = 18$, $m = 19$, $d_i = 2$ for all $i = 1, \ldots, m$,
    \item \textbf{Case 2}: $n = 14$, $m = 10$, $d_i = 2$ for all $i = 1, \ldots, m$.
\end{itemize}

We compare the {\bf tier 3} and {\bf tier 4} approaches.
We define $(L_G^{(d-1)})_{x_t}$ to be the subset of $L_G^{(d-1)}$ consisting of monomials divisible by $x_t$.
Table~\ref{tab:lg} shows, for each degree $d$, the number of elements in $L^{(d-1)}_G$ or in $(L_G^{(d-1)})_{x_t}$, where the latter denotes the subset of $L_G^{(d-1)}$ required for computing $B_d$ using the {\bf tier 4} method. 

From Table~\ref{tab:lg}, we observe that the number of elements to be searched is significantly reduced in {\bf tier 4} compared to {\bf tier 3}. In particular, the reduction becomes more pronounced as the degree increases. Moreover, in Case 1 at $d = 8$ and in Case 2 at $d = 3$, the number of elements to be examined using {\bf tier 4} becomes zero.

These results confirm that {\bf tier 4} greatly reduces the number of elements from $L_G$ that need to be examined when computing $B_d$ from $\widetilde{M}_d$. This reduction contributes substantially to improving the efficiency of the LGB algorithm.
\begin{table}[t]
  \centering
  \caption{Number of elements in $L_G^{(d-1)}$ and $(L_G^{(d-1)})_{x_t}$ at each degree $d$ for {\bf tier 3} and {\bf tier 4}.}
  \label{tab:lg}
  \begin{minipage}[t]{0.48\textwidth}
    \centering
    \begin{tabular}{c|rr}
      \hline
      $d$ & \; {\bf tier 3}  & \; {\bf tier 4}  \\
      \hline
      2  & 0     & 0     \\
      3  & 19    & 4     \\
      4  & 79    & 6     \\
      5  & 268   & 82    \\
      6  & 818   & 234   \\
      7  & 2242  & 1102  \\
      8  & 5320  & 0     \\
      9  & 11134 & 5814  \\
      10 & 21470 & 8398  \\
      \hline
    \end{tabular}
    \subcaption*{Case 1: $n = 18$, $m = 19$, $d_i = 2$}
  \end{minipage}
  \hfill
  \begin{minipage}[t]{0.48\textwidth}
    \centering
    \begin{tabular}{c|rr}
      \hline
      $d$ & \; {\bf tier 3} & \; {\bf tier 4} \\
      \hline
      2  & 0   & 0   \\
      3  & 10  & 0   \\
      4  & 30  & 4   \\
      5  & 69  & 25  \\
      6  & 132 & 48  \\
      7  & 216 & 42  \\
      8  & 306 & 132 \\
      9  & 381 & 207 \\
      10 & 416 & 242 \\
      11 & 425 & 251 \\
      \hline
    \end{tabular}
    \subcaption*{Case 2: $n = 14$, $m = 10$, $d_i = 2$}
  \end{minipage}
\end{table}
\section{Improved LGB and Its Effectiveness}\label{sec:imp}
\subsection{Improved LGB}
In this subsection, we present an improved version of the LGB algorithm by incorporating the methods introduced in Sections~\ref{sec:d}, \ref{sec:M}, and \ref{sec:L}, resulting in a more efficient computation. The improved algorithm is described in Algorithm~3, and we explain the modifications and theoretical foundations that lead to enhanced performance.

The basic structure of Algorithm~3 is the same as that of Algorithm~2, but it is optimized in the following key aspects:
\begin{enumerate}
    \item \textbf{Improved loop termination condition} (lines 2--6, 10): \\
    In Algorithm~2, the loop is terminated when the Hilbert series $H_{R/\langle L_G \rangle}(z)$ computed via the HPS algorithm matches the target Hilbert series $H_{R/\langle F \rangle}(z)$. In Algorithm~3, this termination condition is replaced with a degree-based bound $D$, avoiding repeated HPS computations. The bound $D$ is determined as follows, as discussed in Section~\ref{sec:d}:
    \begin{itemize}
        \item If $n \leq m$, the Hilbert series $H_{R/\langle F \rangle}(z)$ is finite. Then, $D$ is set to one more than the degree of this series.
        \item If $n > m$, the Hilbert series is infinite, but the Macaulay bound provides an upper bound:
        $
        D = \sum_{i=1}^m (d_i - 1) + 1.
        $
    \end{itemize}
    \item \textbf{Reduction in the number of monomials checked for divisibility and reduction in the number of elements in $L_G$ used for those checks} (lines 9, 11--14, 17): \\
    In Algorithm~2, all degree-$d$ monomials are checked for membership in $\langle L_G \rangle$ at each iteration. However, as shown in Sections~\ref{sec:M} and \ref{sec:L}, only a smaller subset of these monomials and generators needs to be considered.
\end{enumerate}

The modifications related to monomial divisibility checks and relevant $L_G$ elements are described below. In Algorithm~3, $\widetilde{B}_d$ denotes the set of monomials of degree $d$ not contained in $\langle L_G \rangle$. Thus, line 9 initializes $\widetilde{B}_0 = \{1\}$, and the loop begins from $d = 1$.

In line 11, $b$ is the largest monomial in $B_d$ (in term order), and $x_t$ is the smallest variable in $b$. As shown in Section~\ref{sub:tier3}, only monomials in
$
\widetilde{M}_d = \{ x_t b \mid b \in \widetilde{B}_{d-1} \}
$
need to be checked for membership. Furthermore, according to Section~\ref{sub:tier4}, only the elements of $L_G$ divisible by $x_t$ need to be considered in this check. Thus, line 14 filters $\widetilde{M}_d$ by removing monomials divisible by $\langle L_G \rangle \cap \langle x_t \rangle$, and adds monomials formed by multiplying $\widetilde{B}_{d-1}$ by variables larger than $x_t$.

In line 17, the $N_d$ largest monomials selected into $L_G$ are removed from $B_d$ to produce $\widetilde{B}_d$.

\begin{algorithm}[t]
    \begin{algorithmic}[1]    
    \caption{{\bf LGB} (improved)}
    \label{alg2}
        \REQUIRE $n \in \mathbb{N}$, $m \in \mathbb{N}$, $\{d_i\}_{i=1}^{m} \in \mathbb{N}^m$
        \ENSURE Leading monomials $L_G$ of minimal Gr\"obner basis for a generic sequence $F = \{f_1,...,f_m\}$ with $n$ variables, $m$ polynomials and $\mathrm{deg}(f_i)=d_i$
        \STATE $H_{R/\langle F \rangle}(z) \leftarrow \left[ \frac{\prod^{m}_{i=1}(1-z^{d_i})}{(1-z)^n} \right]$
        \IF{$n \leq m$}
            \STATE $D \leftarrow \rm{deg}(H_{R/\langle F \rangle}(z))$
        \ELSE
            \STATE $D \leftarrow \sum_{i=1}^{m} (d_i-1)+1$
        \ENDIF
        \STATE $L_G \leftarrow \varnothing$
        \STATE $d \leftarrow 1$
        \STATE $\widetilde{B}_0 \leftarrow \{1\}$
        \WHILE{$d \leq D$}
        \STATE $b \leftarrow \text{ the largest element in term order from } B_d$\\
        \STATE $x_t \leftarrow \text{the smallest variable in } b$\\
        \STATE $\widetilde{M}_d \leftarrow \{x_t b \mid b \in \widetilde{B}_{d-1}\}$
        \STATE $B_d \leftarrow \{m \in \widetilde{M}_d \mid m \notin \langle L_G^{(d-1)} \rangle \cap \langle x_t \rangle \} \cup \{x_i b \mid i \in \{t+1,...,n\}, b \in \widetilde{B}_{d-1}\}$
        \STATE $N_d \leftarrow \#B_d - h_{R/\langle F \rangle}(d)$ 
        \STATE $L_G \leftarrow L_G \cup \{ $the $N_d$ largest elements in term order from $B_d\}$ 
        \STATE $\widetilde{B}_d \leftarrow B_d \setminus \{\text{the } N_d \text{ largest elements in term order from } B_d\}$
        \STATE $d \leftarrow d+1$
        \ENDWHILE
        \RETURN $\{ L_G \}$
    \end{algorithmic}
\end{algorithm}
\subsection{Computation Time and Memory Usage of LGB}

In this subsection, we present a comparison between the proposed LGB algorithm and traditional Gr\"obner basis computation in terms of computation time and memory usage for computing the leading monomials of a minimal Gröbner basis of generic sequences.

The experiments are conducted under the following three cases, each with a different number of equations $m$ relative to the number of variables $n$, and with all polynomial degrees set to $2$:
\begin{itemize}
    \item $m = n - 1$, with $d_i = 2$ for all $i = 1, \ldots, m$,
    \item $m = n$, with $d_i = 2$ for all $i = 1, \ldots, m$,
    \item $m = n + 1$, with $d_i = 2$ for all $i = 1, \ldots, m$.
\end{itemize}

We fix the degrees $d_i$ to be $2$ in all cases to avoid excessive computation time for Gr\"obner basis computations, which would make comparison with the LGB algorithm impractical for larger degrees.

Gr\"obner basis computations were performed using the function \texttt{GrobnerBasis()} in \textsf{Magma} V2.28-4. The input systems were homogeneous polynomial systems over the finite field $\mathbb{F}_{32003}$, with randomly chosen coefficients. While the LGB algorithm is theoretically designed over infinite fields, computations over large finite fields are empirically known to yield similar results. Moreover, infinite field computations tend to involve coefficient swell, making them unsuitable for large-scale practical experiments.

The LGB computations were also implemented in \textsf{Magma} V2.28-4 using Algorithm~3. These experiments were executed on a system powered by an AMD EPYC 7763 CPU with 2~TB of RAM. Both computation time and memory usage were measured.

Figure~1 shows the comparison results of computation time for the three cases. Notably, in the case of $m = n + 1$, the proposed LGB algorithm significantly outperforms the standard Gr\"obner basis computation. This suggests that the benefit of using LGB becomes especially pronounced when $m > n$.

Even in the $m = n$ and $m = n - 1$ cases, the performance gap increases as the problem size grows. For example, in the $m = n - 1$ case, which exhibits the smallest gap among the three, the runtime for $n = 16$ shows that LGB is approximately $2000$ times faster than the Gr\"obner basis computation. The steeper slope in the runtime graph for Gr\"obner basis suggests that the gap will become even larger for bigger parameter sizes.

\begin{figure}[]
    \begin{center}
    \includegraphics[width = 10cm]{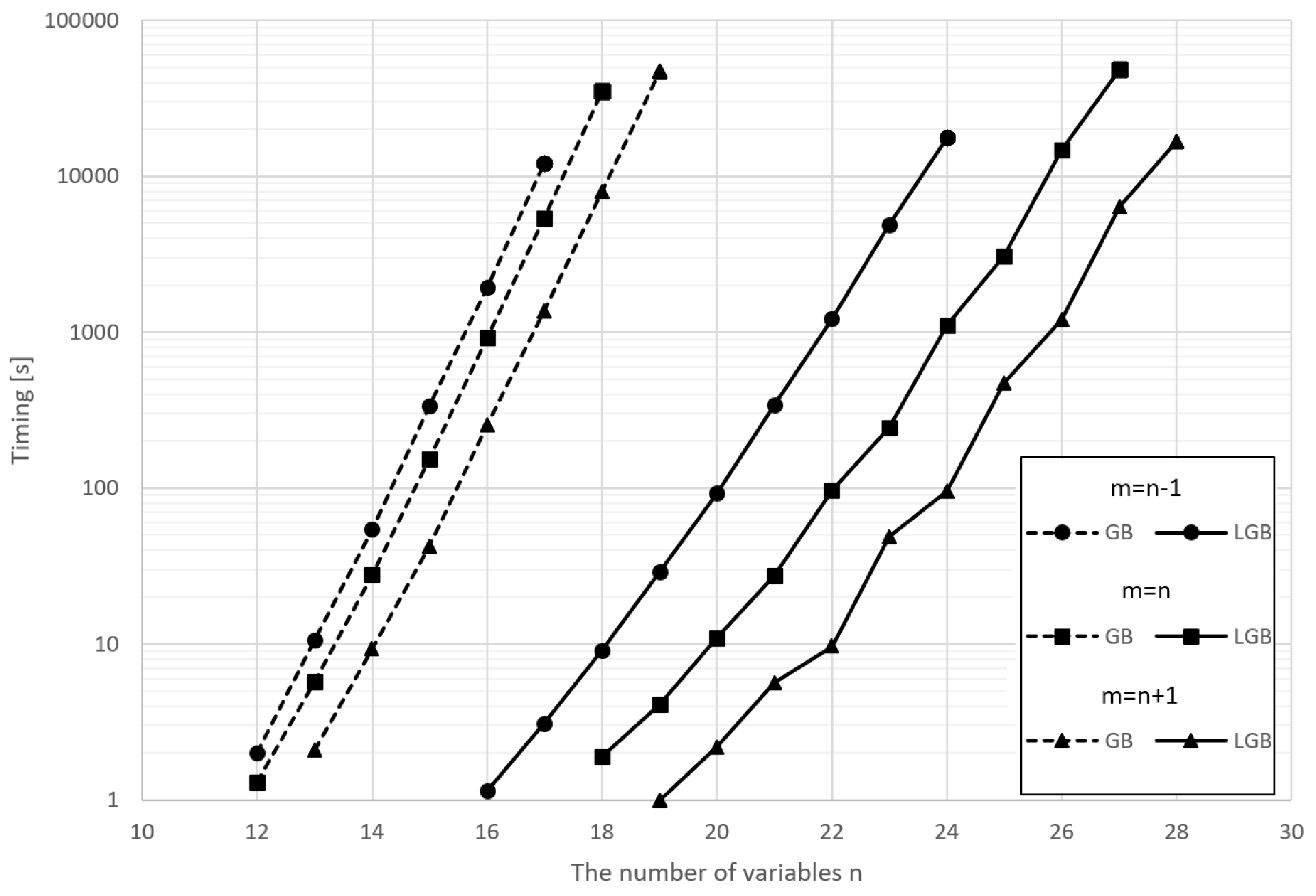}
    \caption{Comparison of timing between Gr\"obner basis computation (GB) and LGB}
    \label{pic1}
\quad \\
    \includegraphics[width = 10cm]{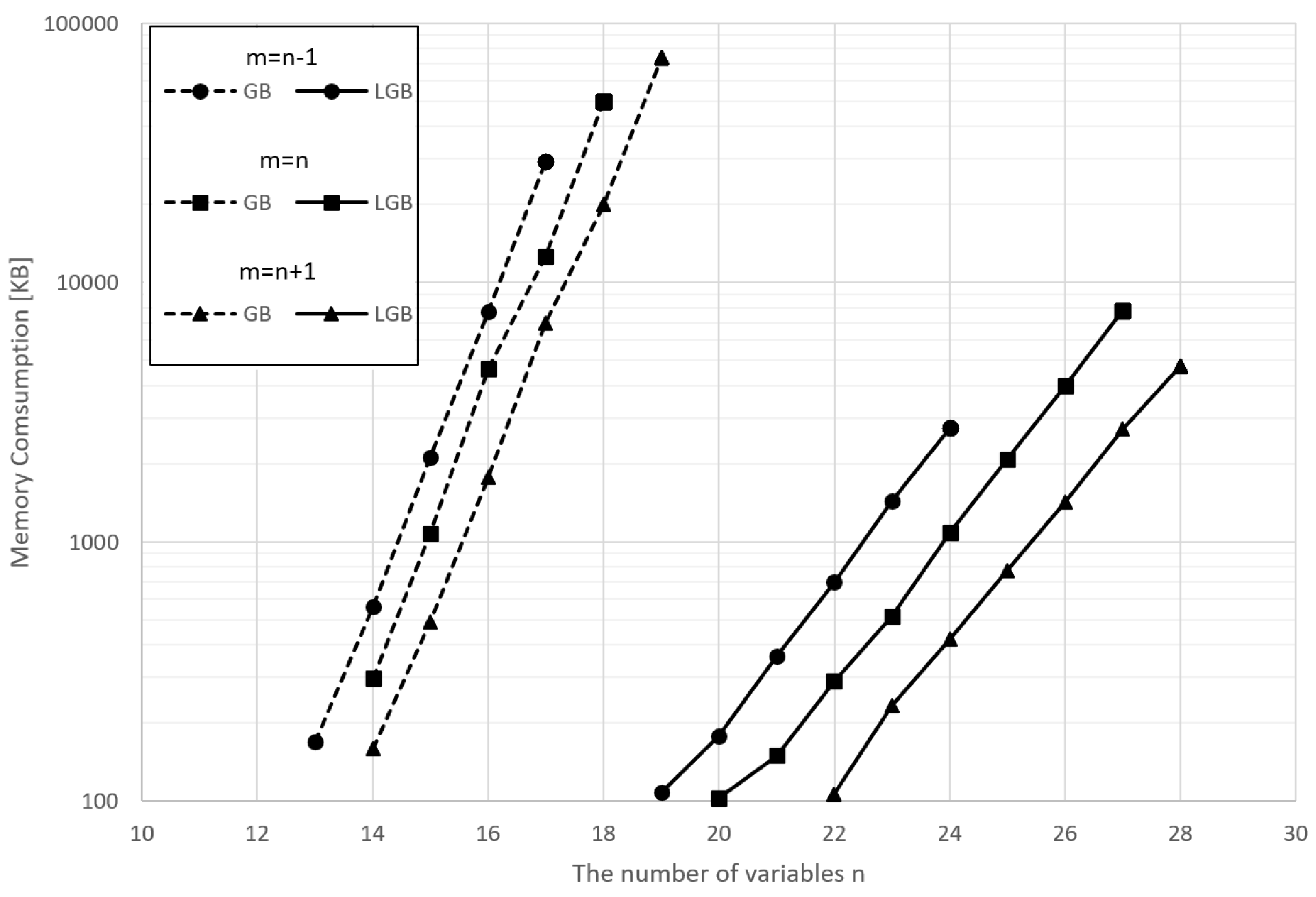}
    \caption{Comparison of memory consumption between Gr\"obner Basis computation (GB) and LGB}
    \label{pic2}
    \end{center}
\end{figure}

\section*{Acknowledgement}
This work was supported by JST K Program Grant Number
JPMJKP24U2, Japan.
This work was also supported by JSPS KAKENHI Grant Number JP25K00140.


\begin{thebibliography}{28}

    
    
    
    
    
    \bibitem{Buch}
    Buchberger, B.:
    Ein Algorithmus zum Auffinden der Basiselemente des Restklassenringes nach einem nulldimensionalen Polynomideal.
    PhD thesis, Universität Innsbruck (1965)
    
    \bibitem{BC}
    Buchberger, B.:
    Criterion for detecting unnecessary reductions in the construction of Gröbner basis.
    In: EUROSAM 1979, LNCS, vol. 72, pp. 3--21. Springer, Heidelberg (1979)

    \bibitem{ChoPark2008}
Cho, Y.H., Park, J.P.: Conditions for generic initial ideals to be almost reverse lexicographic.
J. Algebra \textbf{319}(7), 2761--2771 (2008)
    
    \bibitem{CLO}
    Cox, D., Little, J., O'Shea, D.:
    Ideals, Varieties, and Algorithms.
    Springer, New York (1997)
    
    \bibitem{F4}
    Faugère, J.-C.:
    A new efficient algorithm for computing Gröbner bases (F4).
    J. Pure Appl. Algebra \textbf{139}, 6--88 (1999)
    
    \bibitem{F5}
    Faugère, J.-C.:
    A new efficient algorithm for computing Gröbner bases without reduction to zero (F5).
    In: Proc. ISSAC 2002, pp. 75--83 (2002)
    
    
    \bibitem{Froberg}
    Fröberg, R.:
    An inequality for Hilbert series of graded algebras.
    Math. Scand. \textbf{56}(2), 117--144 (1985)

    \bibitem{FrobergHollman1994}
Fröberg, R., Hollman, J.: Hilbert series for ideals generated by generic forms.
J. Symb. Comput. \textbf{17}(2), 149--157 (1994)
    
    \bibitem{GM}
    Gebauer, R., Möller, H.M.:
    On an installation of Buchberger's algorithm.
    J. Symb. Comput. \textbf{6}, 275--286 (1988)
    
    
    

    \bibitem{HS}
    Hashemi, A., Seiler, W. M.: Dimension and depth dependent upper bounds in polynomial ideal theory. J. Symbolic Comput. \textbf{98}, 47--64 (2020)
    
    \bibitem{MS}
    Moreno-Socías, G.:
    Autour de la fonction de Hilbert-Samuel (escaliers d'idéaux polynomiaux).
    Thèse, École Polytechnique (1991)

    \bibitem{MorenoSocias2003}
Moreno-Socías, G.: Degrevlex Gröbner bases of generic complete intersections.
J. Pure Appl. Algebra \textbf{180}(3), 263--283 (2003)
    
    \bibitem{Pardue}
    Pardue, K.:
    Generic sequences of polynomials.
    J. Algebra \textbf{324}(4), 579--590 (2010)
    
    
    
    \bibitem{HD}
    Traverso, C.:
    Hilbert functions and the Buchberger algorithm.
    J. Symb. Comput. \textbf{22}(4), 355--376 (1996)
    
    
    
    
    
    
    
    
    \end{thebibliography}
\end{document}